\documentclass[journal]{IEEEtran}
\usepackage{amsmath,amssymb,amsthm}
\usepackage{graphicx}
\usepackage{booktabs}
\usepackage{multirow}
\usepackage{cite}
\usepackage{algorithm}
\usepackage{algpseudocode}
\usepackage{hyperref}

\newtheorem{proposition}{Proposition}
\newtheorem{lemma}{Lemma}
\newtheorem{corollary}{Corollary}

\newcommand{\E}{\mathbb{E}}
\newcommand{\Rea}{\operatorname{Re}}
\newcommand{\tr}{\operatorname{tr}}

\begin{document}

\title{Wait or Rotate: Statistical Repair of the Backward Pass in Over-the-Air Split Learning Under Channel Aging}

\author{Xianzhong Li, Yuzhi Yang, Tianhao Guo}

\maketitle

\begin{abstract}
Over-the-air split learning returns the receiver's loss gradient through the reciprocal channel, so terminal motion between the forward and backward transmissions corrupts the backward pass. We repair the aged return at a fixed receive array using only second-order statistics. Two numbers taken from the forward--backward cross-covariance, the alignment reachable by the best common phase and by the best unitary correction, bound the main term of every repair. In the 3GPP tapped-delay-line channels TDL-A/B/C, where every antenna ages by the same scalar, the two numbers coincide, and at a correlation zero all tested fixed corrections stay between 58\% and 62\% accuracy, whereas a return delayed to the next correlation lobe reaches 98.4\%. In 3GPP clustered-delay-line channels a statistical unitary correction stays within 0.02 in mean cross-entropy of ideal channel knowledge in all six tested profile and speed combinations, while common-phase correction falls to 69\% for CDL-A at 3~km/h and the tested delay loses up to 10 points. In a two-cluster model spatial correction is decisive only where the common-phase alignment nearly vanishes (64\% to 95\%). With a ResNet-18 on CIFAR-10, only the delayed return trains at a separable correlation zero (86\% against at most 16\%), and at 30~km/h in a clustered-delay-line channel the unitary correction reaches 80\% where the common phase reaches 60\%. Given geometry, displacement and angular statistics, a port permutation chosen by linear assignment reaches 99.5--99.7\% accuracy for hexagonal and square arrays in oblique motion without pilots or digital feedback during training, where nearest-position matching drops to 89\%.
\end{abstract}


\section{Introduction}

Over-the-air split learning places a layer of a neural network in the wireless channel and uses reciprocity to compute that layer's backward pass in the air. The receiver sends the derivative of its loss back to the transmitter, and the reverse channel applies the required adjoint without a channel estimate~\cite{yang2022jsac,yang2022mlsp,wang2021relay}. The same principle trains photonic and other physical networks in situ~\cite{hermans2015,hughes2018,pai2023,bosch2026unifying}. It relies on the backward channel equalling the forward one. When the terminal moves between the two transmissions, the channel ages~\cite{truong2013aging} and the returned gradient is attenuated and rotated. Once the rotation exceeds $90^\circ$ the mean update climbs the loss, which is the classical failure condition of filtered-reference adaptive control~\cite{morgan1980,boucher1991}. At 28~GHz half a wavelength is about 5~mm, which a pedestrian at 1~m/s covers in 5~ms.

A receiver that knows only channel statistics can repair the aged return in three ways. A common phase rotation of the whole return is exact for a single antenna away from correlation zeros~\cite{wcl}, a spatial correction mixes the returns of different antennas, and a delayed return waits until the correlation has recovered, which is the remedy for a single antenna at a correlation zero~\cite{wcl}. Restoring the antenna position with a movable or reconfigurable antenna~\cite{jsac,artner2020,li2025static} is a fourth option. This paper treats a fixed array, which can only change how it processes the return, and gives a heuristic, based on second-order statistics and applied before training, for choosing among the three repairs.

The choice depends on whether the channel ages separately in space and time, and the standard 3GPP models contain both cases~\cite{3gpp38901}. The non-line-of-sight tapped-delay-line models TDL-A/B/C give every tap a common Jakes temporal process and define no spatial structure. Whatever spatial correlation a simulation attaches to them, the forward--backward cross-covariance is the zero-lag covariance times a scalar, and the array ages exactly like a single antenna. Under these models spatial mixing cannot raise the main-term ceiling above the best common phase, and at a correlation zero neither supplies a conditional mean gradient. The clustered-delay-line (CDL) models attach a Doppler shift to every ray, so the Doppler phase depends on the arrival direction and the aging is not separable. We show that the two families call for opposite repairs, and that two numbers computed from the cross-covariance bound the main term that spatial correction can add and help select a repair, without determining the training outcome on their own.

\emph{Related work.} For channel prediction, Wong and Evans showed that under separable space--time correlation and noiseless observations the joint space--time predictor reduces to per-element temporal prediction~\cite{wong2006}. Corollary~\ref{cor:sep} below is the counterpart of this fact for a fixed correction of the returned gradient. Predictor antennas~\cite{sternad2012,guo2021} and displaced-phase-center processing~\cite{dpca2000} use a trailing element that revisits the position of a leading one. Movable- and fluid-antenna mobility can keep a channel static~\cite{artner2020,li2025static}, fluid-antenna port selection uses temporal and spatial correlation to reduce observations~\cite{zhang2023fas}, and the Hungarian method has planned movable-antenna reconfiguration from environment priors~\cite{ji2026}. These works aim at the channel itself, whereas a returned gradient only needs a positive component along the true gradient. This is why second-order statistics, without channel reconstruction or feedback, can suffice. We rely on the polar factor as the nearest unitary matrix~\cite{higham1986}, on von Neumann's trace inequality~\cite{horn2013}, and on the Hungarian method for the assignment problem~\cite{kuhn1955}. Positive alignment at each step is not enough for convergence. When the aligned component stays above a fixed fraction of the squared gradient, the second moment of the update is controlled, the loss is smooth and bounded below and the step sizes are chosen accordingly, the analysis in~\cite{bottou2018} shows that a weighted average of the squared gradient norms of the resulting biased SGD tends to zero. Companion manuscripts treat a single receive antenna with waiting and few-bit phase correction~\cite{wcl}, movable-antenna position restoration~\cite{jsac}, and per-realization sign laws for port re-indexing on lattices~\cite{lattice}. A further companion report~\cite{lessons} compares unitary correction with truncated inversion near a correlation zero and finds them comparable.

\emph{Contributions.}
\begin{itemize}
\item \emph{Statistical ceilings and a selection rule.} The mean return splits into a main term and a task-dependent residual. The main term reachable by a common phase is the magnitude of the trace of the cross-covariance, and the main term reachable by any correction of unit spectral norm is its nuclear norm. A rule built on the two numbers chooses among phase correction, unitary correction and delayed return before training, and it picked the best repair in 7 of the 8 tested cells where the repairs differ (Sections~\ref{sec:repair} and~\ref{sec:select}).
\item \emph{Separable versus non-separable 3GPP aging.} In TDL-A/B/C channels the two ceilings coincide, and at the tested correlation zero only delayed returns train. In CDL channels a unitary correction computed from statistics comes within 0.02 mean CE of ideal channel knowledge without delay in every tested profile, while the common phase and the tested delayed return each fall behind in some profiles (Section~\ref{sec:3gpp}).
\item \emph{When spatial correction matters.} Over 23 channel configurations of a compact task the gap between the two ceilings ranks the training gain of spatial over phase correction (Spearman $0.73$). A large gain appears where the common-phase main term nearly vanishes, and moderate values of it already suffice for training (Section~\ref{sec:select}). In CIFAR-10 split learning with a ResNet-18 encoder, the separable zero again admits only the delayed return, the unitary correction is the best statistical repair in both clustered channels, and the selection rule picks the best candidate in all three cells (Section~\ref{sec:cifar}).
\item \emph{Zero-feedback port permutation.} Restricted to permutations, the main term becomes a linear assignment problem. Solved on the cross-covariance with a geometric and angular prior, it needs no pilots or digital feedback during training and is realized by a switching network without complex weights. It matches pilot-estimated unitary correction for hexagonal and square arrays in oblique motion and ranks all permutations by trained accuracy, which nearest-position matching does not. With the same prior and normalization it never beats the statistical unitary correction by more than 0.0015 mean CE, and its score identifies configurations, such as a linear array in broadside motion, where pure permutations fall short (Section~\ref{sec:perm}).
\end{itemize}

\section{System Model}\label{sec:model}

\emph{Forward and backward passes.} A transmitter with $N_t$ ports maps an input $x$ to $s=f_W(x)\in\mathbb{C}^{N_t}$. The receiver with $N_r$ ports observes $y=H_f s+n$ and evaluates a loss $L$. Its port-domain seed is the complex gradient $\delta=2\,\partial L/\partial y^*$. The receiver transmits $M\delta$, where $M\in\mathbb{C}^{N_r\times N_r}$ is a fixed linear correction, after power normalization. Through the reciprocal reverse link with phase conjugation the transmitter receives
\begin{equation}\label{eq:return}
h=H_r^H M\delta+w,
\end{equation}
where $H_r$ is the channel at the time of the return and $w$ is reverse noise with $\E[w\mid\mathcal F]=0$. Equation~\eqref{eq:return} idealizes the norm restoration and ignores its quantization. The transmitter's local backward map $\mathcal{J}$, its Jacobian adjoint, gives the update direction $\hat g=\mathcal{J}(h)$. With $H_r=H_f$, $M=I$ and $w=0$ this is the true gradient $g=\mathcal{J}(H_f^H\delta)$.

\emph{Aging statistics.} Between the forward and backward transmissions the receiver moves by $d$ wavelengths. Let the transmit columns of $H_f$ and $H_r$ be zero-mean, proper, and jointly Gaussian with common receive covariances
\begin{equation}
C_{00}=\tfrac{1}{N_t}\E[H_fH_f^H],\qquad C_{10}=\tfrac{1}{N_t}\E[H_rH_f^H],
\end{equation}
independent transmit columns, and $C_{00}\succ0$, and let the innovation $H_r-\E[H_r\mid H_f]$ have zero conditional mean given the entire forward information (input, parameters, $H_f$ and forward noise). Then $\E[H_r\mid H_f]=AH_f$ with $A=C_{10}C_{00}^{-1}$. The angular-field and TDL channels of Section~\ref{sec:setup} are generated this way. The CDL generator sums finitely many rays whose coefficients are shared by the transmit columns, so it does not satisfy this left-multiplicative conditional mean, and for CDL the quantities below are statistical registration measures rather than an exact decomposition of the returned gradient. The scalar trace correlation $\rho=\tr C_{10}/\tr C_{00}$ is the power-weighted average of the per-port temporal correlations, and it equals the single-antenna coefficient when all ports have the same statistics.

\emph{Separable and non-separable aging.} Aging is \emph{separable} at displacement $d$ if $C_{10}=\rho C_{00}$, equivalently $A=\rho I$. In the non-line-of-sight 3GPP models TDL-A/B/C every tap follows a common Jakes process~\cite{3gpp38901}. With any spatial covariance attached to the taps, the aging is separable with $\rho=J_0(2\pi d)$. TDL-D/E add a Ricean first tap with its own Doppler line and are not covered by this statement. In the 3GPP CDL models the channel is a sum of rays with individual angles and Doppler shifts $\nu_r=f_D\sin\theta^{\rm ZOA}_r\cos\phi^{\rm AOA}_r$ for motion along the $x$-axis~\cite{3gpp38901}. The Doppler phase of a ray depends on its direction, so $C_{10}=\sum_r p_re^{j2\pi d\nu_r/f_D}a(r)a(r)^H$ is not a scalar multiple of $C_{00}=\sum_rp_ra(r)a(r)^H$ in general. We quantify non-separability by $\|C_{10}-\rho C_{00}\|_F/\|C_{10}\|_F$.

\section{Repairs and Their Statistical Ceilings}\label{sec:repair}

\subsection{Main Term of the Returned Gradient}
Under the model above, $\E[h\mid\mathcal F]=H_f^HA^HM\delta$, where $\mathcal F$ collects the forward-pass quantities (input, parameters, $H_f$, forward noise and $\delta$). The effective return operator is $B=A^HM$. Let $\bar g$ be the batch-averaged true gradient and $\bar{\hat g}$ the batch average of the conditional-mean return, and assume complex-linear $\mathcal{J}$.

\begin{lemma}[Task-port moment]\label{lem:T}
Let $T=\E_x[H_f\mathcal{J}^H(\bar g)\delta^H]$. Then $\langle\bar g,\bar{\hat g}\rangle=\tr(T^HB)$ and $\tr T=\|\bar g\|^2$.
\end{lemma}
\begin{proof}
With $u=H_f\mathcal{J}^H(\bar g)$, $\langle\bar g,\mathcal{J}(H_f^HB\delta)\rangle=u^HB\delta=\tr(B\delta u^H)$. Averaging over the batch gives the first claim, and $B=I$ gives the trace.
\end{proof}

\begin{proposition}[Main term and residual]\label{prop:main}
Let $c=\|\bar g\|^2/\tr C_{00}$, $T_r=T-cC_{00}$ and $B_0=B-\tfrac{\tr B}{N_r}I$. Then
\begin{gather}
\Rea\langle\bar g,\bar{\hat g}\rangle=\|\bar g\|^2\,m(M)+\Rea\tr(T_r^HB_0),\label{eq:main}\\
m(M)=\frac{\Rea\tr(C_{10}^HM)}{\tr C_{00}},\nonumber
\end{gather}
and $|\Rea\tr(T_r^HB_0)|\le\|T_r\|_F\|B_0\|_F$.
\end{proposition}
\begin{proof}
$\tr T_r=0$ removes the identity part of $B$ from the residual, and $\tr(C_{00}A^HM)=\tr(C_{10}^HM)$.
\end{proof}

The main term $m(M)$ depends only on second-order statistics and on the correction, not on the task. The residual depends on how far the task-port moment $T$ deviates from a multiple of $C_{00}$. The main term decides the sign of the mean descent component whenever $\|\bar g\|^2|m(M)|$ exceeds the residual bound $\|T_r\|_F\|B_0\|_F$. Near zeros, or in parameter states where the residual dominates, the sign can flip, and a guarantee for every loss needs a per-realization analysis such as the one in~\cite{lattice} for lattice port re-indexing.

\subsection{Three Ceilings}
Every repair below uses a correction of unit spectral norm, so it does not raise the norm bound on the returned noise. Define
\begin{equation}\label{eq:ceilings}
R=\frac{|\tr C_{10}|}{\tr C_{00}},\qquad P=\frac{\|C_{10}\|_*}{\tr C_{00}},\qquad G=P-R\ge0.
\end{equation}

\begin{proposition}[Ceilings of the main term]\label{prop:ceil}
\begin{enumerate}
\item Over common phases $M=e^{j\varphi}I$, $\max_\varphi m(M)=R$, attained at $\varphi=\arg\tr C_{10}$.
\item Over all $M$ with $\|M\|_2\le1$, $\max m(M)=P$, attained at the unitary $M_\star=UV^H$ where $C_{10}=U\Sigma V^H$ is a singular value decomposition.
\item Over permutation matrices $\Pi$ with $\Pi_{\pi(i),i}=1$, $m(\Pi)=\sum_i\Rea[C_{10}]_{\pi(i),i}/\tr C_{00}$, and its maximum $S_\star$ is a linear assignment problem with $S_\star\le P$.
\end{enumerate}
\end{proposition}
\begin{proof}
(1) is immediate. (2) is von Neumann's trace inequality, $\Rea\tr(XM)\le\sum_i\sigma_i(X)\sigma_i(M)\le\|X\|_*$ for $\|M\|_2\le1$~\cite{horn2013}, with $X=C_{10}^H=V\Sigma U^H$ and equality at $M=UV^H$, which is the unitary Procrustes problem~\cite{higham1986}. (3) follows from $\tr(C_{10}^H\Pi)=\sum_i\overline{[C_{10}]_{\pi(i),i}}$ and $\|\Pi\|_2=1$.
\end{proof}

$G$ is the gain of the best spatial correction over the best common phase. Equivalently, since $\E\|H_r-MH_f\|_F^2/N_t=\tr C_{11}+\tr C_{00}-2\Rea\tr(C_{10}^HM)$ for unitary $M$, with $C_{11}$ the covariance at the return, $2G\tr C_{00}$ is the reduction in mean-square registration error from the best common phase to the best unitary correction. The ceilings are thus channel-registration quantities, and here they serve to select a repair for the gradient. $G$ has a simple characterization.

\begin{corollary}[When spatial correction cannot raise the main term]\label{cor:sep}
$G=0$ if and only if $C_{10}=e^{j\varphi}S$ for some $\varphi$ and $S\succeq0$. In particular, separable aging gives $G=0$, and at a zero of $\rho$ the main term of every correction with $\|M\|_2\le1$ vanishes.
\end{corollary}
\begin{proof}
$|\tr X|\le\|X\|_*$ holds with equality exactly when $X$ is a common phase times a positive semidefinite matrix. If $C_{10}=\rho C_{00}$, then $\|C_{10}\|_*=|\rho|\tr C_{00}=|\tr C_{10}|$, and $\rho=0$ gives $C_{10}=0$.
\end{proof}

The converse of the separable statement fails: non-separable aging can still have $G=0$, e.g.\ $C_{00}=I$ and $C_{10}=\mathrm{diag}(0.2,0.4)$. Corollary~\ref{cor:sep} settles the design question for separable channels. Under separable aging a fixed array is a single antenna as far as the main term is concerned, and the only statistical remedy at a correlation zero is to change $d$.

\subsection{The Repairs}
\emph{Common phase.} $M=e^{j\arg\tr C_{10}}I$ attains $R$, with the convention $M=I$ when $\tr C_{10}=0$. When $C_{10}$ is rank deficient, the optimal unitary factor below is not unique, and we complete it from a full singular value decomposition. It requires the phase of one scalar, estimated from pilots or computed from a known angular spectrum. A channel can be nearly coherent, $|\rho|$ close to one, and still rotate the return beyond $90^\circ$, and the common phase removes this rotation. It cannot help when $R$ is small.

\emph{Unitary (Procrustes) correction.} $M_\star=UV^H$ attains $P$. With the unweighted predictor $A=U_A\Sigma_AV_A^H$ in place of $C_{10}$ one obtains its polar factor $M=U_AV_A^H$, which maximizes $\Rea\tr(A^HM)$ and is the variant estimated from pilots in our simulations. Both leave the stretch of the channel untouched and remove only its rotation. Truncated inversion performs comparably near a correlation zero~\cite{lessons}, and the unitary factor is used here because it attains the norm-bounded ceiling in closed form.

\emph{Port permutation.} The receiver transmits the seed component of port $i$ from port $\pi(i)$. The best permutation solves the assignment problem of Proposition~\ref{prop:ceil}(3) on the cross-covariance and needs the array geometry, the displacement and the angular statistics. The displacement follows from the terminal speed and the exchange timing, as in antenna selection on a movable platform~\cite{jsac}, and the angular statistics are treated as known over the training interval. Given these, training uses no pilots and no digital feedback beyond the analog return itself, and when they are unavailable the pilot-estimated unitary correction applies. A simpler relative is the partial shift used by predictor antennas~\cite{sternad2012,guo2021} and antenna selection on a movable platform~\cite{jsac}, which sends the seed of each port to the next port along the motion and leaves the boundary port silent. The assignment instead keeps every port and applies to any geometry, and its optimum need not be a shift. For three elements at $d=1.0$ in an isotropic field it reverses the port order, with main term $0.46$ against $0.22$ for the best cyclic shift.

\emph{Delayed return.} The receiver holds the return until a later displacement $d'$ and corrects it there~\cite{wcl}. In this paper $d'$ is the first later negative minimum of $\Rea\rho$ in the window $(d,d+3]$ wavelengths, located on a 3001-point grid and refined, and the return is sign-flipped. If the window holds no negative minimum, this repair is unavailable. For TDL the first negative minimum of $J_0(2\pi d)$ is at $d\approx0.610$, and since $\rho$ is real there this recovers $|\rho(d')|$. When $\rho$ is complex, waiting instead for the largest $|\rho|$ on the 3000 grid points of the same window and correcting the common phase there attains the largest common-phase main term on the grid, which is at least the sign-flip main term at every grid point. The window spans 308~ms at 3~km/h and 31~ms at 30~km/h at 3.5~GHz. Section~\ref{sec:3gpp} tests this variant separately, while Algorithm~\ref{alg:select} and its evaluation use the sign-flip delay. The delay costs latency rather than pilots.

\emph{Selection rule.} Algorithm~\ref{alg:select} combines the ceilings with thresholds $r_0=0.5$ and $g_0=0.1$, set before the evaluation of Section~\ref{sec:select}.

\begin{algorithm}[t]
\caption{Statistical repair selection}\label{alg:select}
\begin{algorithmic}[1]
\Require $C_{10}(d)$, $C_{00}$, the trace correlation $\rho(d')$ on the window $(d,d+3]$, thresholds $r_0=0.5$ and $g_0=0.1$
\State compute $R$, $P$, $G=P-R$ by~\eqref{eq:ceilings}
\If{$R\ge r_0$} \State common phase $M=e^{j\arg\tr C_{10}}I$
\ElsIf{$G\ge g_0$}
  \State spatial: solve the assignment of Prop.~\ref{prop:ceil}(3) for $S_\star$
  \If{only port switching is available and $S_\star\ge0.9P$} \State $M=\Pi_\star$
  \Else \State $M=M_\star$ (statistics or pilots)
  \EndIf
\Else \State delay the return to the first later negative minimum of $\Rea\rho$ in the window, and flip its sign (unavailable if none exists)
\EndIf
\end{algorithmic}
\end{algorithm}

\subsection{Two Scattering Clusters}\label{sec:twocluster}
Two scattering clusters make $G$ explicit. Take a uniform linear array of $N$ elements with half-wavelength spacing moving along its axis, and two equal-power point clusters at angles $\theta_1,\theta_2$ to the axis with normalized steering vectors $u,v$. Let $q=|u^Hv|$ and let $t=\pi d(\cos\theta_1-\cos\theta_2)$ be half the difference of the two Doppler phases. A rank-two singular value computation gives, with $G=0$ in the degenerate case $P=R=0$,
\begin{gather}
P=\sqrt{1-q^2\sin^2t},\qquad R=|\cos t|,\nonumber\\
G=\frac{(1-q^2)\sin^2t}{P+R}\quad(P+R>0),\label{eq:twocluster}
\end{gather}
Spatial correction helps when the clusters are both angularly resolvable ($q<1$) and dephased by motion ($t$ away from multiples of $\pi$). For small separations $G$ grows as the fourth power of the difference of the direction cosines, so it rises smoothly from zero without threshold behaviour. For $N=3$, $d=0.383$ and $\theta_1=35^\circ$, $G$ reaches $0.15$ at a separation of $35.8^\circ$ for point clusters and $36.1^\circ$ for clusters of uniform $\pm5^\circ$ width, where $R=0.82$. For the $\pm5^\circ$ clusters at a separation of $36^\circ$, $R=0.83$ at $d=0.383$, whereas $d=1.0$ drives $R$ to $0.03$ with $P=0.89$. Point clusters are a limiting case with singular $C_{00}$, and the ceilings remain defined because they need no inverse.

\section{Simulation Setup}\label{sec:setup}

\emph{Task.} This section describes the compact-task experiments of Sections~\ref{sec:3gpp}--\ref{sec:select}, and the CIFAR-10 setting is given in Section~\ref{sec:cifar}. The compact-task experiments jointly train a complex $2\times3$ encoder and a receiver on a binary geometry task. Three independent $\pm1$ symbols are mapped to the two transmit ports, and the receiver predicts the agreement of two of them from tanh-compressed port energies through a logistic head. The runs of Sections~\ref{sec:3gpp}, \ref{sec:perm} and~\ref{sec:select} use 3840 training examples in 120 Adam steps of batch 32 and 8192 test examples. The earlier isotropic and narrow-cluster runs of Section~\ref{sec:iso} use the same task and budgets, except that their CSI reference trains on 2666 examples in 84 updates to match the RF budget (Table~\ref{tab:budget}). The encoder norm and receiver weights are projected onto fixed bounds. Forward, reverse and pilot noise standard deviations are $0.1$, $0.03$ and $0.1$. Every training example draws an independent channel pair at a displacement $d$ that is fixed within a run.

\emph{Channels.} (i) \emph{Angular-field models} are zero-mean proper Gaussian fields over 128 or more plane-wave directions (isotropic, narrow von Mises clusters, or two finite-width clusters) evaluated at the forward and displaced array positions. (ii) \emph{3GPP CDL-A/B/C}~\cite{3gpp38901} use all clusters with 20 rays each, the official angles, delays, powers and angular spreads, random ray coupling fixed per seed, and single-polarized omnidirectional elements. The carrier is 3.5~GHz, the receive array lies along the $y$-axis and moves along the $x$-axis (broadside motion), eight tones at offsets $-3.5,-2.5,\dots,3.5$~MHz sample a 300~ns delay spread, and the forward--backward interval is 40~ms, i.e.\ $d=0.389$ at 3~km/h and $3.89$ at 30~km/h. Aging depends on the displacement in wavelengths, so at 28~GHz the same displacements are reached after about 5~ms at either speed. (iii) \emph{3GPP TDL-A/B/C} use the standard tap delays and powers with a common Jakes process per tap. TDL specifies no angles, so a fixed receive spatial correlation (a mixture of an isotropic component and a von Mises cluster) is attached, and the return is placed at the exact first Jakes zero $d=j_{0,1}/(2\pi)\approx0.3827$ with the same array orientation and carrier. The arrays are 3- or 4-element uniform linear arrays with half-wavelength spacing, a 7-element hexagonal array and a $3\times3$ square array, and the transmitter has two ports, a half-wavelength ULA along the $y$-axis in the CDL runs and independent columns in the TDL and angular-field runs.

\emph{Arms.} We compare no correction, the common phase (from statistics or from 64 forward/return pilot pairs), the unitary correction (the Procrustes solution $M_\star$ from statistics or the polar factor of a 64-pilot estimate of $A$), per-antenna phase (the phases of the diagonal of the estimated $A$), port permutation by linear assignment on the statistical $C_{10}$, the delayed return (to a later negative lobe of the real trace correlation, with a sign flip, or to the largest $|\rho|$ in the same window with a common phase correction), and CSI, the exact gradient $\mathcal{J}(H_f^H\delta)$, as the digital reference. Two return normalizations are used, and the reverse noise of standard deviation $0.03$ is added at the radio before any rescaling. A norm-field arm sends $M\delta/\|M\delta\|$ and restores the norm from a 16-bit field, so its effective gradient noise scales with $\|\delta\|$. A fixed-bound arm sends $M\delta/K$ with the public bound $K=8\sqrt{N_r}$ and multiplies the received vector by $K$, so its effective noise is $K$ times the radio noise. Since the receiver weights are clipped to $[-4,4]$, each seed coordinate of the tanh-energy receiver satisfies $|\delta_i|\le8$, so $\|M\delta\|\le K$ for every unitary $M$ and the bound is never exceeded. In the 3GPP runs of Table~\ref{tab:3gpp} and in the earlier runs of Section~\ref{sec:iso}, pilot-based and statistical arms use the norm field, while the permutation and the partial shift use the fixed bound, and the matched comparison of Section~\ref{sec:perm} adds a fixed-bound statistical unitary arm. A zero seed is sent as zero. In the oblique-motion and two-cluster runs all physical arms use the fixed bound. Pilots, correction bits and RF uses are counted per arm (Table~\ref{tab:budget}).

\emph{Statistics.} Results are averaged over 8 independent seeds, and methods share inputs, channels, noise and initialization within a seed. Intervals are paired seed bootstrap 95\% intervals. We report test accuracy and cross-entropy (CE). Comparisons against a CE margin of $0.02$ use mean differences and are descriptive.

\begin{table}[t]\centering
\caption{Digital overhead per training run (geometry task)}\label{tab:budget}
\footnotesize\setlength{\tabcolsep}{4pt}
\begin{tabular}{lcc}\toprule
Arm & RF uses & control bits\\\midrule
Common phase, 3-ULA (64 pilots) & 8000 & 73\,728\\
Unitary correction, 3-ULA (64 pilots) & 8000 & 74\,016\\
CSI, 3-ULA (2666 examples) & 7998 & 170\,624\\
Port permutation, 7/9 elements & 7680 & 0\\
Unitary correction, 7 / 9 elements & 8256 / 8384 & 30\,240 / 39\,456\\\bottomrule
\end{tabular}
\\[2pt]\parbox{0.95\columnwidth}{\scriptsize Ledgers count RF uses, pilots and control bits at the physical layer. On the 3-element array, pilot-based arms count 16 bits per example for the norm field and 32 bits per complex coefficient for the 64 pilot pairs on the $2\times3$ channel, plus 288 bits for a $3\times3$ correction. The CSI reference counts 3 RF uses and 64 control bits per training example. Port permutation uses a public fixed amplitude bound instead of a per-sample norm field and assumes known geometry, displacement and angular statistics, whose acquisition is not included in these counts.}
\end{table}

\section{Results}\label{sec:results}

\subsection{Separable and Non-Separable 3GPP Channels}\label{sec:3gpp}
Table~\ref{tab:3gpp} contrasts the two 3GPP families with a 4-element array, with all arms trained in one run on shared seeds. In the three TDL profiles at the first Jakes zero, $C_{10}=0$ to machine precision and $R=P=0$. No fixed correction can then supply a mean direction. The common phase and the statistical unitary correction coincide with the uncorrected return ($58$--$61\%$), and the pilot-estimated unitary correction and the permutation stay between $58\%$ and $62\%$. Delaying the return to the negative lobe at $d=0.61$ and flipping its sign reaches $98.4\%$ in all three profiles, within $0.002$ CE of CSI.

The CDL channels have non-separability residuals of $0.71$--$0.95$, and $P$ exceeds $R$ in every profile. The unitary correction computed from statistics reaches $98.4$--$99.1\%$, within $0.02$ CE of CSI in all six cells. Using instead the polar factor of $A$ estimated from 64 forward/return pilot pairs per training run changes the mean CE by at most $0.0006$ and the accuracy by at most $0.01$ points, so in these cells exact statistics are not required for a unitary repair. The common phase fails in CDL-A at 3~km/h ($68.6\%$, $0.52$ CE above the unitary correction) and is $0.017$--$0.027$ CE worse in three further cells. The tested delayed return, which flips the sign at the first later negative minimum of $\Rea\rho$, is worse than the unitary correction in all six cells, by up to $10.0$ points (CDL-A, 30~km/h). A separate run with 8 new seeds adds the stronger delay that waits for the largest $|\rho|$ in the same window and corrects the common phase there. It improves on the sign-flip delay by up to $0.243$ CE (CDL-A, 30~km/h) and in CDL-A at 3~km/h matches the unitary correction within $0.0005$ CE, at the cost of $30.6$~ms of extra delay. Even against the better of the two delays, the unitary correction without waiting has the lower CE in five of the six cells (Table~\ref{tab:t8}). The latency of waiting is not counted in these per-update comparisons. The TDL and CDL rows use different spatial marginals. The TDL failure nevertheless holds for any spatial marginal by Corollary~\ref{cor:sep}, and every CDL cell has $P>R$, which is the room the unitary correction uses.

\begin{table*}[t]\centering
\caption{3GPP channels, 4-element ULA, broadside motion: ceilings and test accuracy (\%), mean over 8 seeds. CDL: 40~ms interval. TDL: return at the first Jakes zero, 3.8~ms at 100~Hz Doppler}\label{tab:3gpp}
\footnotesize\setlength{\tabcolsep}{5pt}
\begin{tabular}{lcc|ccccccc}\toprule
Channel, km/h & $R$ & $P$ & None & Phase & Unitary & Unitary (64 pilots) & Permutation & Delay & CSI\\\midrule
TDL-A & 0.00 & 0.00 & 58.2 & 58.2 & 58.2 & 58.4 & 60.2 & 98.4 & 98.5\\
TDL-B & 0.00 & 0.00 & 61.2 & 61.2 & 61.2 & 58.8 & 62.0 & 98.4 & 98.5\\
TDL-C & 0.00 & 0.00 & 59.7 & 59.7 & 59.7 & 61.8 & 58.0 & 98.4 & 98.5\\
\midrule
CDL-A, 3 & 0.25 & 0.65 & 49.2 & 68.6 & 98.5 & 98.5 & 71.5 & 94.7 & 98.5\\
CDL-A, 30 & 0.16 & 0.28 & 48.9 & 97.8 & 98.4 & 98.4 & 98.2 & 88.5 & 98.5\\
CDL-B, 3 & 0.58 & 0.78 & 49.5 & 98.8 & 98.8 & 98.8 & 93.4 & 98.8 & 98.8\\
CDL-B, 30 & 0.10 & 0.22 & 49.1 & 98.1 & 98.8 & 98.8 & 75.3 & 98.0 & 98.8\\
CDL-C, 3 & 0.41 & 0.64 & 48.6 & 98.6 & 99.1 & 99.1 & 85.3 & 98.8 & 99.1\\
CDL-C, 30 & 0.19 & 0.23 & 49.1 & 98.6 & 98.7 & 98.7 & 64.2 & 97.9 & 99.1\\
\bottomrule
\end{tabular}
\\[2pt]\parbox{0.95\textwidth}{\scriptsize Phase and Unitary are computed from the exact statistics. TDL rows delay the return to the negative lobe at $d=0.61$ (6.1~ms) with a sign flip. CDL rows delay the return to the first later negative minimum of the real trace correlation, with a sign flip. Permutation uses a fixed public amplitude bound, and the other physical arms restore a 16-bit norm field.}
\end{table*}

\begin{table*}[t]\centering
\caption{Matched run with 8 new seeds, CDL cells (km/h): largest correlation magnitude $|\rho(d^\star)|$ in the window, extra delay (ms) of the sign-flip and phase-corrected delays, and mean test CE of each arm. Fixed: statistical unitary correction with the fixed bound $K$}\label{tab:t8}
\footnotesize\setlength{\tabcolsep}{4pt}
\begin{tabular}{lccc|cccccc}\toprule
Channel & $|\rho(d^\star)|$ & ms, flip & ms, phase & Delay, flip & Delay, phase & Unitary & Fixed & Permutation & CSI\\\midrule
CDL-A, 3 & 0.69 & 22.5 & 30.6 & 0.184 & 0.069 & 0.069 & 0.071 & 0.561 & 0.068\\
CDL-A, 30 & 0.16 & 10.5 & 0.7 & 0.322 & 0.079 & 0.071 & 0.075 & 0.082 & 0.068\\
CDL-B, 3 & 0.58 & 23.3 & 0.1 & 0.071 & 0.070 & 0.069 & 0.069 & 0.226 & 0.069\\
CDL-B, 30 & 0.16 & 8.6 & 6.3 & 0.076 & 0.076 & 0.073 & 0.076 & 0.422 & 0.069\\
CDL-C, 3 & 0.40 & 25.7 & 0.1 & 0.066 & 0.070 & 0.055 & 0.057 & 0.295 & 0.053\\
CDL-C, 30 & 0.19 & 9.0 & 1.4 & 0.083 & 0.077 & 0.070 & 0.076 & 0.555 & 0.053\\
\bottomrule
\end{tabular}
\end{table*}

\subsection{Correlation Zeros in an Isotropic Field}\label{sec:iso}
For a 3-element array in an isotropic field at the first zero of $J_0$, $d=0.383$, the trace correlation is $-0.0008$ while $P=0.598$. The displacement carries each element close to the former position of its neighbour, so the cross-covariance has small diagonal and large adjacent entries. With 64 pilots, unitary correction improves on common-phase correction by $0.276$ CE ($[0.253,0.305]$, 8/8 seeds) averaged over $d\in\{0.38,0.40,0.42\}$, and stays within the $0.02$ CE margin of delaying the return to the next extremum at $d=0.610$ (difference $-0.0013$, $[-0.0033,0.0009]$) without the extra delay. The partial shift by one port, which leaves the first port silent, reaches $96.8$--$96.9\%$ with no pilots, against $96.9\%$ for unitary correction and $97.2\%$ for CSI. Correcting each antenna's phase separately does the opposite. Averaged over the same displacements it is worse than the common phase by $0.270$ CE ($[0.218,0.321]$) and falls to $52.6$--$64.2\%$, because the diagonal of $C_{10}$ is what vanishes at the zero. The Procrustes solution weighted by $C_{00}$ stays within the same margin of the unweighted polar factor ($-0.0002$, $[-0.0006,0.0002]$).

Away from zeros the common phase suffices. For narrow clusters at $35^\circ$ the trace correlation keeps $|\rho|>0.9$ while its phase passes $90^\circ$. Over the four such cells uncorrected training reaches $51.7\%$, common-phase correction $92.5\%$ and CSI $93.0\%$ (CE gain $0.454$, $[0.433,0.479]$). The magnitude of the correlation therefore does not decide the direction of the return, whereas the main term does, and here the common phase attains it. Likewise at the isotropic reference displacements $d\in\{0.55,0.70\}$, away from the zero, neither spatial correction shows a stable gain over the common phase.

\subsection{Zero-Feedback Port Permutation}\label{sec:perm}
For arrays that are not linear, or motion not along the array axis, no shift maps the array onto itself. Table~\ref{tab:perm} uses a 7-element hexagonal array and a $3\times3$ square array at $d=0.383$ in an isotropic field, with motion at $15^\circ$, $37^\circ$ and $73^\circ$. The assignment-selected permutation reaches $99.54$--$99.68\%$, within $0.0018$ CE of the 64-pilot unitary correction in every cell, with zero pilots and zero digital feedback. A greedy rule that matches each port to the nearest displaced position, without the global assignment, works at $15^\circ$ but drops to $91.4\%$ and $89.3\%$ for the hexagonal array at $37^\circ$ and $73^\circ$. The assignment score is the main term of Proposition~\ref{prop:main}, and across all $7!$ and $9!$ permutations, each trained separately, its Spearman correlation with mean test accuracy is $0.95$--$0.97$. Fig.~\ref{fig:perm} shows two cells. Accuracy rises steeply once the score turns positive, and the greedy match of the hexagonal array at $37^\circ$ lands on the slope rather than on the plateau. In every cell the score-optimal permutation is within $0.06$ points of the best in-sample accuracy over all permutations. A separate run with 8 new seeds compares permutation and statistical unitary correction under the same prior and the same fixed bound. On the oblique arrays of Table~\ref{tab:perm} the unitary correction is never worse than the permutation by more than $0.0004$ mean CE and has the lower mean CE in five of the six cells. In the two-cluster cells of Section~\ref{sec:select} the permutation is never better by more than $0.0015$ CE and is worse by up to $0.27$ CE ($16.1$ points at $d=1.0$ without separation). These are cell means over seeds. With matched information the CE advantage of the permutation is therefore at most marginal (the largest observed accuracy advantage is $0.21$ points), and its appeal is implementation, a switching network with no complex weights. Permutation also depends on the geometry of the motion. In the CDL profiles of Table~\ref{tab:3gpp} the 4-element array moves broadside, so no element lands near the former position of another. The best permutation then captures only $7$--$67\%$ of the spatial ceiling $P$, and permutation accuracy ranges from $64.2\%$ to $98.2\%$, below the statistical unitary correction by $0.011$--$0.546$ CE. The gap is not an artifact of signalling. With the same fixed bound, the permutation still trails the unitary correction by $0.156$ to $0.491$ CE in the five cells with a low score ratio, and the fixed-bound and norm-field unitary corrections differ by at most $0.006$ CE. The ratio of the assignment score to $P$ is available before training and is at most $0.351$ in the five cells where the permutation loses more than $0.1$ CE. Likewise in the two-cluster model at $d=1.0$ the permutation falls to $75.1\%$ without cluster separation and to $87.3\%$ at $36^\circ$ (Fig.~\ref{fig:dial}). A low ratio shows only that pure permutations miss the ceiling, since a common phase combined with a permutation, or another sparse unitary matrix, may still reach it. Zero-feedback permutation therefore suits arrays whose displacement nearly maps the aperture onto itself, and the unitary correction covers the remaining cases.

\begin{table}[t]\centering
\caption{Oblique motion, isotropic field, $d=0.383$: test accuracy (\%)}\label{tab:perm}
\footnotesize\setlength{\tabcolsep}{3pt}
\begin{tabular}{lcccccc}\toprule
Array, direction & None & Greedy & Assignment & Unitary & CSI\\\midrule
Hex 7, $15^\circ$ & 59.6 & 99.39 & 99.54 & 99.59 & 99.56\\
Hex 7, $37^\circ$ & 62.1 & 91.38 & 99.55 & 99.57 & 99.56\\
Hex 7, $73^\circ$ & 65.6 & 89.25 & 99.56 & 99.56 & 99.56\\
Square 9, $15^\circ$ & 57.3 & 99.68 & 99.68 & 99.68 & 99.71\\
Square 9, $37^\circ$ & 67.0 & 99.61 & 99.64 & 99.65 & 99.71\\
Square 9, $73^\circ$ & 66.0 & 99.62 & 99.65 & 99.68 & 99.71\\\bottomrule
\end{tabular}
\\[2pt]\parbox{0.95\columnwidth}{\scriptsize Assignment and greedy use the exact geometry, displacement, direction and isotropic spectrum. Unitary: polar factor of a 64-pilot estimate. Unitary is pilot-estimated here, and the matched statistical comparison is reported in Section~\ref{sec:perm}. The assignment accuracies also appear in the report~\cite{lessons}.}
\end{table}

\begin{figure*}[t]\centering
\includegraphics[width=\textwidth]{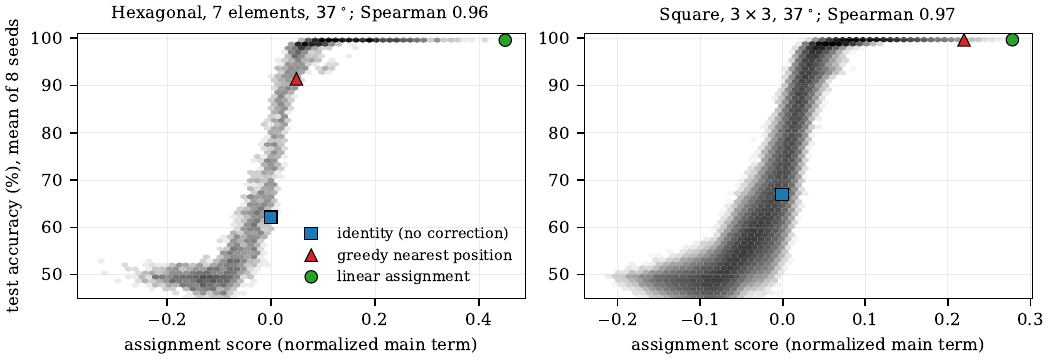}
\caption{Every receive-port permutation trained separately ($7!$ and $9!$ candidates, 8 seeds each), isotropic field, $d=0.383$, motion at $37^\circ$. The horizontal axis is the main term $m(\Pi)$ of Proposition~\ref{prop:main}, and shading counts permutations on a log scale. Linear assignment maximizes the score without pilots.}\label{fig:perm}
\end{figure*}

\subsection{Predictive Value of the Ceilings}\label{sec:select}
A dedicated run evaluates the ceilings and Algorithm~\ref{alg:select} over 23 configurations, the 14 two-cluster cells of Section~\ref{sec:twocluster} ($d\in\{0.383,1.0\}$, separations $0^\circ$--$60^\circ$, 3-element array, axial motion, all physical arms with a fixed amplitude bound) and the 9 cells of Table~\ref{tab:3gpp}. The training gain is the CE of statistical common-phase correction minus that of statistical unitary correction.

\begin{figure*}[t]\centering
\includegraphics[width=\textwidth]{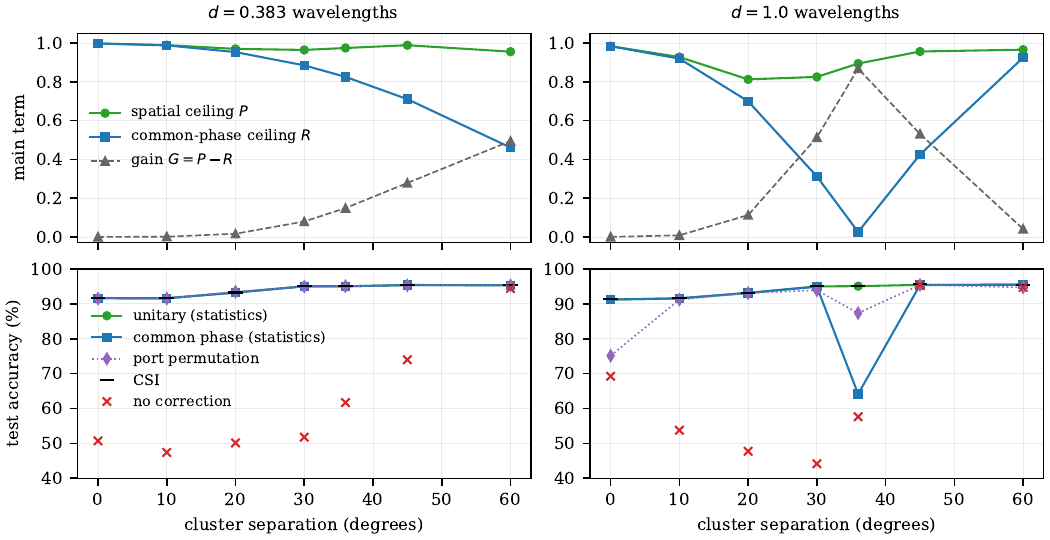}
\caption{Two equal-power clusters ($\pm5^\circ$ wide, the first at $35^\circ$) seen by a 3-element ULA moving along its axis. Top: ceilings of the main term. Bottom: test accuracy after training with corrections computed from the statistics, mean over 8 seeds. At $d=0.383$ the spatial gain $G$ passes $0.15$ near $36^\circ$ without any change in training. At $d=1.0$ and $36^\circ$ the common-phase ceiling nearly vanishes and only spatial correction trains.}\label{fig:dial}
\end{figure*}

\emph{A large common-phase main term suffices.} In the 10 cells with $R\ge0.6$ the gain is below $0.0015$ CE. At $d=0.383$ (Fig.~\ref{fig:dial}, left), $G$ rises from $0$ to $0.49$ and crosses $0.15$ near $36^\circ$, while in all seven cells the two corrections stay within $0.0015$ CE of each other and within $0.003$ CE of CSI. The two-cluster value $G=0.15$ of Section~\ref{sec:twocluster} marks no change in training.

\emph{Spatial correction is decisive only where the common phase nearly vanishes.} At $d=1.0$ and $36^\circ$, where $R=0.026$ and $P=0.89$, common-phase correction trains to $64.2\%$ and statistical unitary correction to $95.1\%$, a gain of $0.457$ CE ($[0.396,0.511]$). A low common-phase ceiling combined with a high spatial ceiling is not sufficient for a large gain. At $30^\circ$, with $R=0.31$ and $P=0.83$, the gain is $0.004$ CE, contrary to the gain of at least $0.10$ CE we had predicted for $R\le0.35$ and $P\ge0.8$. In CDL, CDL-A at 3~km/h fails with $R=0.25$ whereas CDL-B at 30~km/h trains with $R=0.10$. Between these cases the ceilings do not determine the outcome. For the angular-field cells the remainder is the task-dependent residual of Proposition~\ref{prop:main}, and for the CDL cells it also includes the part of the return that the left-multiplicative model of Section~\ref{sec:model} does not describe.

\emph{Ranking and selection.} Across the 23 cells the Spearman correlation between $G$ and the gain is $0.73$ (seed-resampling interval $[0.43,0.77]$), against $0.36$ for $-R$. Algorithm~\ref{alg:select} selected the repair with the lowest CE in 7 of the 8 cells where the candidate repairs differ by more than $0.02$ CE. In the remaining cell, CDL-C at 30~km/h ($R=0.19$, $G=0.04$), the rule delays the return (CE $0.102$) while the unitary correction ($0.073$) and the common phase ($0.077$) are better. The run also shows that the statistical unitary correction was never worse than the common phase by more than $0.00022$ CE in any of the 23 cells. When the full correction is affordable, a simpler rule identified after this run, which corrects spatially unless $G$ is below $0.01$ and then delays if $R$ is also small, would have chosen correctly in all 8 cells.

\subsection{CIFAR-10 Split Learning}\label{sec:cifar}
To test whether these conclusions survive a large model, a ResNet-18 encoder at the transmitter maps each CIFAR-10 image to 8 channel uses on 4 transmit antennas, and the 4-element terminal, moving broadside, combines the forward signal by MMSE with forward CSI and classifies it with a two-layer MLP. The backward pass is returned per sample with a restored norm, forward and reverse noise standard deviations are $0.1$ and $0.03$, and training uses SGD with momentum $0.9$, weight decay $5\times10^{-4}$, a one-cycle learning rate peaking at $0.05$, batch 128 and 40 epochs with random crops and flips. One channel pair at a single 3.5~GHz frequency serves each minibatch, and the channel field is redrawn every 100 updates. The three cells are the TDL channel at the first Jakes zero, CDL-A at 3~km/h and CDL-B at 30~km/h, each with 4 seeds after 2 separate check seeds. The delayed return changes only the displacement of the return channel and adds no wall-clock time in these runs. The ResNet encoder is not complex-linear and the channel field is reused over 100 updates, so Proposition~\ref{prop:main} does not apply exactly and these runs are empirical.

\begin{table*}[t]\centering
\caption{CIFAR-10 split learning with a ResNet-18 encoder, 4-element terminal in broadside motion: mean final test accuracy (\%) over 4 seeds, and in parentheses the mean epoch at which 80\% is first reached (-- if some seed does not reach it within 40 epochs)}\label{tab:cifar}
\footnotesize\setlength{\tabcolsep}{6pt}
\begin{tabular}{lcc|ccccc}\toprule
Channel & $R$ & $P$ & None & Phase & Unitary & Delay & CSI\\\midrule
TDL, first Jakes zero & 0.00 & 0.00 & 14.2 & 15.8 & 15.7 & 86.1 & 93.7\\
 & & & (--) & (--) & (--) & (29.5) & (7.0)\\
CDL-A, 3 km/h & 0.25 & 0.65 & 91.9 & 91.0 & 92.7 & 92.7 & 93.9\\
 & & & (17.0) & (18.0) & (12.3) & (11.5) & (6.8)\\
CDL-B, 30 km/h & 0.10 & 0.22 & 21.7 & 60.1 & 80.0 & 64.9 & 93.9\\
 & & & (--) & (--) & (--) & (--) & (7.0)\\
\bottomrule
\end{tabular}
\end{table*}

Table~\ref{tab:cifar} confirms that separable aging at a correlation zero leaves only the delayed return. The uncorrected, common-phase and unitary returns stay at $14$--$16\%$ accuracy, close to chance for ten classes, whereas the delayed return reaches $86.1\%$, $7.6$ points below CSI, at a lobe where the correlation magnitude is $0.40$. In both CDL cells the unitary correction is the best statistical repair. In CDL-B at 30~km/h, where both ceilings are small, it lifts accuracy from $60.1\%$ with the common phase to $80.0\%$ and beats the delayed return by $15.1$ points, but remains $13.8$ points below CSI. In CDL-A at 3~km/h the large model tolerates the aged return, all statistical repairs end within $3$ points of CSI, and the unitary correction reaches $80\%$ about 6 epochs earlier than the common phase.

The large and compact setups disagree on when the common phase suffices. It fails in CDL-A at 3~km/h for the compact task and nearly suffices there for CIFAR-10, and the reverse holds in CDL-B at 30~km/h. The two setups differ in task, transmit array, receiver, channel reuse and optimizer, so this reversal shows that the full system configuration decides when the common phase suffices, not the learning task alone. Two conclusions hold in both: at a separable correlation zero only a delayed return trains, and the unitary correction is never worse than the common phase by more than $0.00022$ CE in the compact task or by more than $0.1$ points in the CIFAR-10 cells. Algorithm~\ref{alg:select} selected the best of its candidate repairs in all three CIFAR-10 cells, and so did the simpler rule of Section~\ref{sec:select}, which was stated before these runs. These cells reuse channel profiles of the compact runs, so they test the rules on a new task rather than on new propagation conditions.

\section{Conclusion}
A fixed array that returns gradients through an aging channel can correct the common phase, correct spatially, or wait, and the trace and nuclear norm of the forward--backward cross-covariance tell these repairs apart. When the two ceilings coincide, as in the 3GPP channels TDL-A/B/C where every antenna ages by the same scalar, no correction has a mean to recover at a correlation zero and the array must wait. When they differ, as in 3GPP CDL channels, a unitary correction computed from statistics stays within 0.02 CE of ideal channel knowledge without delay on the compact task, and on CIFAR-10 it has the highest mean final accuracy among the statistical repairs while remaining 13.8 points below CSI in CDL-B at 30~km/h. On the compact task, a port permutation found by linear assignment matches the pilot-estimated unitary correction with zero feedback when the displacement nearly maps the aperture onto itself. Spatial correction is decisive where the common-phase main term nearly vanishes, and on the compact task the gap between the ceilings ranks its training gain elsewhere. Which repair suffices between these cases depends on the system configuration, as the compact and CIFAR-10 experiments show, whereas across the tested cell means the unitary correction was no more than 0.00022 CE worse than the common phase on the compact task and no more than 0.1 points worse on CIFAR-10. Estimating the statistics online, with the cost of that estimation, is the next step.

\bibliographystyle{IEEEtran}
\bibliography{refs}

\end{document}